\documentclass{article}

\title{CRYPTANALYSIS OF A KEY EXCHANGE PROTOCOL\\
BASED ON A CONGRUENCE-SIMPLE SEMIRING ACTION}

\author{A. Otero S\'anchez$^{1}$, J.A. L\'opez Ramos$^{2}$ \\
$^{1-2}$ Department of Mathematics, \\
University of Almer\'{\i}a, \\
04120 Almer\'{\i}a, \\
Spain \\
aos073@inlumine.ual.es, jlopez@ual.es}

\newenvironment{proof}{%
  \par\noindent%
  \textsf{\underline{Proof}}%
    \par%
  }%
  {\par\hfill $\blacksquare$%
}%

\usepackage{amssymb}
\usepackage{amsmath}

\newtheorem{theorem}{Theorem}[section]

\newtheorem{lemma}[theorem]{Lemma}
\newtheorem{rem}[theorem]{Remmark}
\newtheorem{proposition}[theorem]{Proposition}
\newtheorem{definition}[theorem]{Definition}

  {\par\hfill $\blacksquare$%
}%
\newtheorem{example}[theorem]{Example}
\newtheorem{alg}[theorem]{Algorithm}

\begin{document}
\maketitle
\begin{abstract}
We show that a previously introduced key exchange based on a congruence-simple semiring action is not secure by providing an attack that reveals the shared key from the distributed public information for any of such semirings.\end{abstract}


\section{Introduction}	

Public key cryptography rises up with Diffie-Hellman foundational paper, \cite{diffiehellman}, where they define a secure key exchange based on the so-called Discrete Logarithm Problem over finite fields. This problem can be extended to any cyclic group as it was shown by Miller and Koblitz in \cite{miller} and \cite{koblitz}, respectively, where they used the group structure of the set of points of an elliptic curve to have an elliptic curves version of the Diffie-Hellman protocol still widely used. In 2007, Maze et al. in \cite{maze} defined a general setting that provides, as particular cases, the previous Diffie-Hellman key exchange protocols both on finite fields and elliptic curves. They define the Semigroup Action Problem as follows:

\medskip

{\it Let $(G,\cdot )$ be a semigroup, $S$ be a set,  and let $\varphi :G\times S \rightarrow S$ be an action of $G$ on $S$, i.e. $\varphi$ satisfies that $\varphi (g\cdot h,s)=\varphi (g,\varphi (h,s))$. The Semigroup Action Problem establishes that given $x \in S$ and $y\in \varphi (G,x)$, find $g\in G$ such that $\varphi (g,x)=y$. }

\medskip

Maze et al. show that using this problem, it is possible to define a Diffie-Hellman type key exchange by using an action of an abelian semigroup on a set, extending the classical protocols. As an example of its possible use, they propose an action given by a semiring. 

\medskip

Let us fix a finite semiring $R$, not embeddable in a field and not necessarily commutative. Given such a semiring, consider $C$, the center of $R$. 
 Let $n$ denote an arbitrary positive integer. For $M \in $ Mat$_n(R)$, denote by $C[M]$ the abelian sub-semiring generated by $M$, i.e., the semiring of polynomials in $M$ 
 with coefficients in $C$. Let $M_1, M_2 \in$ Mat$_n(R)$ and consider the following action:

$$(C[M_1] \times C[M_2]) \times \rm{Mat}_n(\it{R})  \longrightarrow \rm{Mat}_n(\it{R})$$
$$\big((p(M_1), q(M_2)), X\big)  \longmapsto p(M_1) \cdot X \cdot q(M_2)$$

If Alice and Bob wish to exchange a key using this semiring action as a basis, then they agree on a finite semiring $R$ with nonempty center $C$, not embeddable into a field. They choose a positive integer $n$ and matrices $M_1, M_2 \in$ Mat$_n(R)$ and act as follows:

\begin{enumerate}

\item Alice chooses polynomials $p_a, q_a \in C[t]$ and computes $A = p_a(M_1)\cdot S\cdot q_a(M_2)$. She sends $A$ to Bob.

\item Bob chooses polynomials $p_b, q_b \in C[t]$ and computes $B = p_b(M_1)\cdot S\cdot q_b(M_2)$. He sends $B$ to Alice.
\end{enumerate}

The common key computed by both is then

$$p_a(M_1)Bq_a(M_2) = p_a(M_1)p_b(M_1)Sq_b(M_2)q_a(M_2) = p_b(M_1)Aq_b(M_2).$$

They assert that, for cryptographic purposes, it is important that the involved semirings are congruence-simple to avoid a Pohlig-Hellman type reduction of the Semigroup Action Problem given that any congruence relation on $R$ yields a projection of the SAP instance onto a quotient semiring, from which one may gain information about the solution to the original instance. Thus, they use this justification in order to work in congruence-simple semirings. They illustrate the use of this action with an example on a semiring with 6 elements constructed with the classification obtained by Zumbr\"agel in \cite{zumbragel}. Later in \cite{steinwandt}, the authors make a cryptanalysis of the example by solving a system of equations based on the operation tables of this particular semiring. 

Our aim in this paper is to show that it is possible to develop a general attack on this protocol for any congruence-simple semiring just by assuming a bound of the degrees of the polynomials that are used in the definition of the protocol.

\section{Algebraic background}

\begin{definition}
A set $R$ with two internal operations $+$ and $\cdot$ is called a semiring if both operations are associative and verify that $a\cdot (b+c)=a\cdot b+a\cdot c$ and $(b+c)\cdot a=b\cdot a+c\cdot a$ for every $a,b,c\in R$. We say that $S$ is additively (resp. multiplicatively) commutative in case addition (resp. multiplication) satisfies commutativity. We say that $R$ is commutative in case both operations satisfy commutativity. 
\end{definition}

\begin{definition} A congruence relation on a semiring $R$ is an equivalence relation $\sim$ such that

$$a\sim b \Rightarrow \left\{ \begin{array}{ccc} 
a+c & \sim & b+c \\
c+a & \sim & c+b \\
a\cdot c & \sim & b\cdot c \\
c\cdot a & \sim & c\cdot b \\
\end{array}\right.$$

\noindent for every $a,b,c\in R$. A semiring $R$ that admits no congruence relations other than the trivial
ones, $id_R$ and $R \times R$, is said to be congruence-simple.
\end{definition}

The following result, due to C. Monico, establishes a set of conditions, among which, one is necessarily satisfied by every  additively commutative congruence-simple semiring.

\begin{theorem} (\cite[Theorem 4.1]{monico}\label{necessary}
Let $R$ be a finite, additively commutative, congruence-simple semiring. Then one of the following holds:
\begin{enumerate}
\item $\vert R\vert \leq 2$.
\item $R\cong \rm{Mat}_n(\mathbb{F}_q)$, for some finite field $\mathbb{F}_q$ and some $n \geq 1$.
\item $R$ is a zero multiplication ring of prime order.
\item $R$ is additively idempotent, i.e., $x+x=x$ for every $x\in R$.
\item $R$ contains an absorbing element $\infty$, i.e. $x\cdot \infty=\infty \cdot x=\infty$ for every $x\in R$ and $R+R=\{ \infty \}$. 
\end{enumerate}
\end{theorem}

In \cite[Proposition 3.1]{zumbragel} J. Zumbr\"agel showed an important property concerning the addition of these semirings.

\begin{proposition} \label{zum}
Let $R$ be a congruence-simple semiring which is not a ring. Then the addition $(R, +)$ is idempotent.
\end{proposition}

The next result follows immediately by the definition of idempotent element. 

\begin{lemma}
Let $R$ be a semiring such that $x+x=x$ for every $x\in R$ and let $S=$Mat$_n(R)$ Then the following statements hold
\begin{enumerate}
    \item For every $x,y,z\in R$
    \begin{equation}\label{eqn:propSemianillo1}
        x+y=z \Rightarrow x+z=y+z=z.
    \end{equation}    

    \item   For every $M\in S$ 
    \begin{equation} \label{eqn:propSemianillo2}
         M+M=M 
    \end{equation}

    \item  For every $X,Y, Z\in  S$ 
    \begin{equation} \label{eqn:propSemianillo3}
         X+Y=Z \Rightarrow X+Z=Y+Z=Z  
    \end{equation}

    \item  For any finite set $I$ 
    \begin{equation} \label{eqn:propSemianillo4}
         \sum_{i\in I}M_i=Z \Rightarrow M_i+Z=Z \quad \forall M_{i},Z \in S, \forall i\in I. 
    \end{equation}
   
\end{enumerate}
\end{lemma}

Let $R$ be a semiring, and let $C=\{ c\in R: cx=xc, \  \forall x\in R\}$ i.e. the multiplicative center of $R$. Let us denote by $C[x]$ the set of polynomials in $x$ with coefficients in $C$. Then $C_m[x]$ will be the subset of $C[x]$ whose elements have degree less than or equal to $m$. 

We now prove a result that will allow us to achieve our aim.  

\begin{theorem}\label{main}
Let $R$ be a finite, additively idempotent semiring and let the matrices $M_1,S,M_2, A \in $Mat$_n(R)$. Assume that there exist two polynomials  $\varphi , \psi \in C_{m}[x]$  such that $\varphi(M_1)S\psi(M_2)=A $. Let $W$ be the set

$$W=\big\{(x,a,b) \in C \times \{0,1,\dots ,m\}\times \{ 0,1, \dots ,m\} : x M_1^a S M_2^b+A=A\big\}$$.
 
Then
\begin{equation}
    F_{S}[X,Y]=\sum_{(k,i,j)\in W} k X^i S Y^j
\end{equation}
verifies that 
\begin{equation}
    F_{S}[M_1,M_2]=A
\end{equation}
\end{theorem}

\begin{proof} Let $\varphi(x)=\sum_{i\in I} v_{i} x^i$ and $\psi(x)=\sum_{i\in J} w_{i} x^i$ for $I,J \subseteq \{1,...,m\}$, $v_i,w_i \in C$. Then  

    $$\varphi(M_1)S\psi(M_2) = \sum_{(i,j) \in I \times J} v_i w_j M_1^i S M_2^j \quad \mbox{and} \quad \varphi(M_1)S\psi(M_2)=A$$  
 
Note that $v_i w_j \in C$. By (\ref{eqn:propSemianillo2}) in the preceding Lemma, we have that 

     $$\sum_{(i,j) \in I\times J} v_i w_j M_1^iSM_2^j=A \Rightarrow \sum_{(i,j) \in I\times J} v_i w_j M_1^iS M_2^j + A = A$$

If we apply now  (\ref{eqn:propSemianillo4}) of the  previous Lemma, given that $I\times J$ is a finite set, we get that
$v_a w_b M_1^a S M_2^b+A=A \quad \forall(a,b) \in I\times J$. Therefore, $(v_a w_b,a,b) \in W \quad \forall(a,b) \in I\times J$. Then, if we denote $U=\{(v_a w_b,a,b) : (a,b)\in I \times J\}$, we have $U \subset W$.

\medskip

Now, set $K=W \setminus U$. Then, by the definition of $W$, we get that $F_{S}[M_1,M_2]=\sum_{(x,i,j)\in W} x M_1^iSM_2^j =  \sum_{(x,i,j)\in U} x M_1^iSM_2^j +  \sum_{(x,i,j)\in K} x M_1^iSM_2^j = \\ A+\sum_{(x, i,j)\in K} x M_1^iSM_2^j = A$

\end{proof} 

\begin{rem}
    The set $W$ used in Theorem \ref{main} can be reduced to a smaller set using a greedy approach as described in \cite{steinwandt}. As was pointed out in that work, this method is specially useful in case the semiring $R$ has a zero element, $0$, an identity element $1$ and the center is given by the set $\{0, 1\}$.
\end{rem}

We end this section by considering the case of semirings $R$ with $\vert R\vert \leq 2$. In  \cite[Theorem 14.1]{elbashir}, the authors show a classification of such semirings which are commutative.

\begin{theorem}\label{order2}
Let $R$ be a congruence-simple additively commutative semiring such that $\vert R\vert \leq 2$. Then $R$ is isomorphic to one of the following:

\medskip

\begin{tabular}{| c  c | c  c |}
\hline
 &  &  &  \\
\begin{tabular}{ c | c  c }
$(T_1,+)$ & 0 & 1 \\ \hline
0 & 0 & 0 \\
1 & 0 & 0 \\
\end{tabular} & 
\begin{tabular}{ c | c  c }

$(T_1,\cdot)$ & 0 & 1 \\ \hline
0 & 0 & 0 \\
1 & 0 & 0 \\
\end{tabular} & 
\begin{tabular}{ c | c  c }

$(T_2,+)$ & 0 & 1 \\ \hline
0 & 0 & 0 \\
1 & 0 & 0 \\
\end{tabular} & 
\begin{tabular}{ c | c  c }

$(T_2,\cdot)$ & 0 & 1 \\ \hline
0 & 0 & 0 \\
1 & 0 & 1 \\
\end{tabular}\\ 
&  &  &  \\
\hline
  &  &  &  \\

\begin{tabular}{ c | c  c }

    $(T_3,+)$ & 0 & 1 \\ \hline
    0 & 0 & 0 \\
    1 & 0 & 1 \\
\end{tabular} & 
\begin{tabular}{ c | c  c }

    $(T_3,\cdot)$ & 0 & 1 \\ \hline
    0 & 0 & 0 \\
    1 & 0 & 0 \\
\end{tabular} & 
\begin{tabular}{ c | c  c }

    $(T_4,+)$ & 0 & 1 \\ \hline
    0 & 0 & 0 \\
    1 & 0 & 1 \\
\end{tabular} & 
\begin{tabular}{ c | c  c }

    $(T_4,\cdot)$ & 0 & 1 \\ \hline
    0 & 1 & 1\\
    1 & 1 & 1 \\
\end{tabular}\\ 

&  &  &  \\
\hline
  &  &  &  \\

\begin{tabular}{ c | c  c }
$(T_5,+)$ & 0 & 1 \\ \hline
0 & 0 & 0 \\
1 & 0 & 1 \\
\end{tabular} & 
\begin{tabular}{ c | c  c }

$(T_5,\cdot)$ & 0 & 1 \\ \hline
0 & 0 & 1 \\
1 & 1 & 1 \\
\end{tabular} & 
\begin{tabular}{ c | c  c }

$(T_6,+)$ & 0 & 1 \\ \hline
0 & 0 & 0 \\
1 & 0 & 1 \\
\end{tabular} & 
\begin{tabular}{ c | c  c }

$(T_6,\cdot)$ & 0 & 1 \\ \hline
0 & 0 & 0 \\
1 & 0 & 1 \\
\end{tabular}\\ 

 &  &  &  \\
\hline
  &  &  &  \\

\begin{tabular}{ c | c  c }

$(T_7,+)$ & 0 & 1 \\ \hline
0 & 0 & 1 \\
1 & 1 & 0 \\
\end{tabular} & 
\begin{tabular}{ c | c  c }

$(T_7,\cdot)$ & 0 & 1 \\ \hline
0 & 0 & 0 \\
1 & 0 & 0 \\
\end{tabular} & 
\begin{tabular}{ c | c  c }

$(T_8,+)$ & 0 & 1 \\ \hline
0 & 0 & 1 \\
1 & 1 & 0 \\
\end{tabular} & 
\begin{tabular}{ c | c  c }

$(T_8,\cdot)$ & 0 & 1 \\ \hline
0 & 0 & 0\\
1 & 0 & 1 \\
\end{tabular}\\ 
 &  &  &  \\
\hline

\end{tabular}

\end{theorem}

The following Lemma shows that any other additively commutative semiring with exactly two elements is necessarily additively idempotent. 

\begin{lemma}\label{order2idempotent}
There is no additively non-commutative and non-additively idempotent semiring with exactly 2 elements. 
\end{lemma}

\begin{proof}
    Let $(R=\{0,1\}, +, \cdot)$ an additively non-commutative and non-additively idempotent semiring.  Then the addition should be given by one of the following possibilities. 
    
\begin{center}
    \begin{tabular}{ c  |  c  |  c}
         &    \\
        \begin{tabular}{ c | c  c }
            Op1 & 0 & 1 \\ \hline
            0 & 1 & $a$ \\
            1 & $b$ & 0 \\
        \end{tabular}
        & 
        \begin{tabular}{ c | c  c }
            Op2 & 0 & 1 \\ \hline
            0 & 0 & $a$ \\
            1 & $b$ & 0 \\
        \end{tabular} 
        &
        \begin{tabular}{ c | c  c }
            Op3 & 0 & 1 \\ \hline
            0 & 1 & $a$ \\
            1 & $b$ & 1 \\
        \end{tabular}

    \end{tabular}
\end{center}

\noindent where $a,b \in R$ and $a\not = b$ since otherwise, $R$ would be additively commutative. 

Let us check first the case given by Op1 and assume that $a=0$ and $b=1$. Then the table looks like 

\begin{center}
    \begin{tabular}{ c | c  c }
        Op1 & 0 & 1 \\ \hline
        0 & 1 & 0 \\
        1 & 1 & 0 \\
    \end{tabular}
\end{center}

\noindent The following equalities show that the ring is not additively associative, which is a contradiction. 

\[1+(0+1) = 1 + 0 = 1 \]
\[(1+0)+1 = 1 + 1 = 0 \]

Following analogous reasonings for $a=1$ and $b=0$ as well as in the remaining cases for Op2 and Op3 show similar contradictions. 
\end{proof}

\section{Cryptanalysis of the protocol}

Let us focus now in our aim, that is, revealing the common key shared by both parties acting in the protocol proposed by Maze et al. in \cite{maze} in the setting given by a congruence-simple semiring. We recall that both parties agree on such semiring $R$ and matrices  $S, M_1, M_2 \in$ Mat$_n(R)$. Then they choose their own private keys, given by two polynomials $p_A(x), q_A(x)$ and $p_B(x),q_B(x)$ in $C[x]$, being $C$ the center of $R$ and exchange $A = p_a(M_1)\cdot S\cdot q_a(M_2)$. and $B = p_b(M_1)\cdot S\cdot q_b(M_2)$. Our aim is to get the shared key given by 

$$p_a(M_1)Bq_a(M_2) = p_b(M_1)Aq_b(M_2)$$

\noindent from the public information, $S, M_1, M_2, A, B$. To this end, and taking into account Theorem \ref{necessary} and Proposition \ref{zum} we have to consider different possibilities on the semiring $R$.

\subsection{$R$ is additively idempotent}.

\noindent We are in the case of Theorem \ref{main}, so we can apply it to the matrices $M_1, M_2$ and $B$. However, given that the polynomials $p_A(x), q_A(x), p_B(x)$ and $q_B(x)$ are private and thus, their degrees are unknown, let us fix $m$ an upper bound for all of them that will depend on the computational capabilities of the attacker. This is necessary in order that the computation of the polynomial $F_S[X,Y]$ is a finite process. Thus, the attack operates in the following way:

\begin{alg}\label{attack}

Input: $S, M_1, M_2, A, B$.

\medskip

\hspace{2cm} Output: Alice and Bob's shared key. 

\begin{enumerate}
\item Compute $W=\big\{(k,a,b) \in C \times \{0,1,\dots ,m\}^2 : k M_1^a S M_2^b+A = A\big\}$.

\item Compute the polynomial in three variables $F[X,Y,Z]=\displaystyle{\sum_{(k, i,j)\in W}} k X^iZY^j$.

\item Compute the shared key given by $F[M_1,M_2,B]$. 

\end{enumerate}
\end{alg}

\begin{lemma}
Algorithm \ref{attack}  provides the shared key from the public information. 
\end{lemma}

\begin{proof}
The output of the algorithm provides the shared key given that
\begin{multline}
    F[M_{1},M_{2},S] = \sum_{(k,i,j) \in W} k M_{1}^i S M_{2}^j = \sum_{(k,i,j) \in W} k M_{1}^i p_{b}(M_{1}) S q_{b}(M_{2}) M_{2}^j = \\ = \sum_{(k,i,j) \in W} k p_{b}(M_{1}) M_{1}^i  S  M_{2}^j q_{b}(M_{2})= p_{b}(M_{1}) \big( \sum_{(k, i,j) \in W} k M_{1}^i S M_{2}^j  \big) q_{b}(M_{2}) \\
= p_{b}(M_{1}) F[M_{1},M_{2},S] q_{b}(M_{2}) =p_{b}(M_{1}) A q_{b}(M_{2}) 
\end{multline}

where we have used that, by Theorem \ref{main} $F[M_1,M_2,S]=A$. 

\end{proof}

\begin{lemma}
Let $n\times n$ be the size of the matrices used in the key exchange protocol. If we assume that $m$ is the chosen upper bound for the degree of the private polynomials, then the number of operations of Algorithm \ref{attack} is $O(m^2n^3)$. 
\end{lemma}

\begin{proof}
The cost for building the polynomial $F[X,Y,Z]$ is basically the cost of determining the set $W$ of Theorem \ref{main}. 

Firstly, note that in \cite[Remark 3.9]{zumbragel} it is said that if the semiring $R$ is finite and additively idempotent, then the center $C$ of $R$ is the set $\{0,1\}$, so it is only necessary to check if $M_1^a S M_2^b + A =A$ for $a,b \in \{0,1,\cdots,m\}$. 

Secondly, the cost of determining $W$ is given by $O(m^3)$ total matrix products, and $O(m^2)$ matrix additions and comparisons. However, we can compute $M_1^{a+1} S M_2 ^b$ and $M_1^a S M_2 ^{b+1}$ using $M_1^a S M_2 ^b$. Therefore, it is possible to reduce the number of matrix products to $O(2m)$ by storing all those products already calculated. 

Now given that, in the worst case, the number of operations in the base semiring $R$ for a matrix product of two square matrices of order $n$  is $O(n^3)$, then we get that the number of operations in the base semiring is $O(m^2n^3)$. 
\end{proof}

\begin{example}
The polynomial $F[X,Y,Z]$ that provides the shared key of \cite[Example 5.13]{maze} is given by the following

\begin{equation}
\begin{split}
F [X,Z,Y]= & X^{1} Z Y^{0}+X^{2} Z Y^{0}+X^{3} Z Y^{0}+X^{4} Z Y^{0}+X^{6} Z Y^{0}+X^{7} Z Y^{0}+X^{9} Z Y^{0}+\\ X^{1} Z Y^{5}+
& X^{2} Z Y^{5}+X^{3} Z Y^{5}+X^{4} Z Y^{5}+X^{6} Z Y^{5}+X^{7} Z Y^{5}+X^{9} Z Y^{5}+X^{1} Z Y^{7}+\\ X^{2} Z Y^{7}+
& X^{3} Z Y^{7}+X^{4} Z Y^{7}+X^{6} Z Y^{7}+X^{7} Z Y^{7}+X^{9} Z Y^{7}+X^{1} Z Y^{8}+X^{2} Z Y^{8}+\\ X^{3} Z Y^{8}+
& X^{4} Z Y^{8}+X^{6} Z Y^{8}+X^{7} Z Y^{8}+X^{9} Z Y^{8}+X^{1} Z Y^{9}+X^{2} Z Y^{9}+X^{3} Z Y^{9}+\\ X^{4} Z Y^{9}+
& X^{6} Z Y^{9}+X^{7} Z Y^{9}+X^{9} Z Y^{9}+X^{1} Z Y^{10}+X^{2} Z Y^{10}+X^{3} Z Y^{10}+X^{4} Z Y^{10}+\\ X^{6} Z Y^{10}+
& X^{7} Z Y^{10}+X^{9} Z Y^{10}+X^{1} Z Y^{11}+X^{2} Z Y^{11}+X^{3} Z Y^{11}+X^{4} Z Y^{11}+X^{6} Z Y^{11}+\\ X^{7} Z Y^{11}+
& X^{9} Z Y^{11}+X^{1} Z Y^{12}+X^{2} Z Y^{12}+X^{3} Z Y^{12}+X^{4} Z Y^{12}+X^{6} Z Y^{12}+X^{7} Z Y^{12}+\\ X^{9} Z Y^{12}+
& X^{1} Z Y^{13}+X^{2} Z Y^{13}+X^{3} Z Y^{13}+X^{4} Z Y^{13}+X^{6} Z Y^{13}+X^{7} Z Y^{13}+X^{9} Z Y^{13}+\\ X^{1} Z Y^{14}+
& X^{2} Z Y^{14}+X^{3} Z Y^{14}+X^{4} Z Y^{14}+X^{6} Z Y^{14}+X^{7} Z Y^{14}+X^{9} Z Y^{14}+X^{1} Z Y^{15}+\\ X^{2} Z Y^{15}+
& X^{3} Z Y^{15}+X^{4} Z Y^{15}+X^{6} Z Y^{15}+X^{7} Z Y^{15}+X^{9} Z Y^{15}+X^{1} Z Y^{16}+X^{2} Z Y^{16}+\\ X^{3} Z Y^{16}+
& X^{4} Z Y^{16}+X^{6} Z Y^{16}+X^{7} Z Y^{16}+X^{9} Z Y^{16}+X^{1} Z Y^{17}+X^{2} Z Y^{17}+X^{3} Z Y^{17}+\\ X^{4} Z Y^{17}+
& X^{6} Z Y^{17}+X^{7} Z Y^{17}+X^{9} Z Y^{17}+X^{1} Z Y^{18}+X^{2} Z Y^{18}+X^{3} Z Y^{18}+X^{4} Z Y^{18}+\\ X^{6} Z Y^{18}+
& X^{7} Z Y^{18}+X^{9} Z Y^{18}+X^{1} Z Y^{19}+X^{2} Z Y^{19}+X^{3} Z Y^{19}+X^{4} Z Y^{19}+X^{6} Z Y^{19}+\\ X^{7} Z Y^{19}+
& X^{9} Z Y^{19}+X^{1} Z Y^{20}+X^{2} Z Y^{20}+X^{3} Z Y^{20}+X^{4} Z Y^{20}+X^{6} Z Y^{20}+X^{7} Z Y^{20}+\\ X^{9} Z Y^{20}+
& X^{1} Z Y^{21}+X^{2} Z Y^{21}+X^{3} Z Y^{21}+X^{4} Z Y^{21}+X^{6} Z Y^{21}+X^{7} Z Y^{21}+X^{9} Z Y^{21}+\\ X^{1} Z Y^{22}+
& X^{2} Z Y^{22}+X^{3} Z Y^{22}+X^{4} Z Y^{22}+X^{6} Z Y^{22}+X^{7} Z Y^{22}+X^{9} Z Y^{22}+X^{1} Z Y^{23}+\\ X^{2} Z Y^{23}+
& X^{3} Z Y^{23}+X^{4} Z Y^{23}+X^{6} Z Y^{23}+X^{7} Z Y^{23}+X^{9} Z Y^{23}+X^{1} Z Y^{24}+X^{2} Z Y^{24}+\\ X^{3} Z Y^{24}+
& X^{4} Z Y^{24}+X^{6} Z Y^{24}+X^{7} Z Y^{24}+X^{9} Z Y^{24}+X^{1} Z Y^{25}+X^{2} Z Y^{25}+X^{3} Z Y^{25}+\\ X^{4} Z Y^{25}+
& X^{6} Z Y^{25}+X^{7} Z Y^{25}+X^{9} Z Y^{25}+X^{1} Z Y^{26}+X^{2} Z Y^{26}+X^{3} Z Y^{26}+X^{4} Z Y^{26}+\\ X^{6} Z Y^{26}+
& X^{7} Z Y^{26}+X^{9} Z Y^{26}+X^{1} Z Y^{27}+X^{2} Z Y^{27}+X^{3} Z Y^{27}+X^{4} Z Y^{27}+X^{6} Z Y^{27}+\\ X^{7} Z Y^{27}+
& X^{9} Z Y^{27}+X^{1} Z Y^{28}+X^{2} Z Y^{28}+X^{3} Z Y^{28}+X^{4} Z Y^{28}+X^{6} Z Y^{28}+X^{7} Z Y^{28}+\\ X^{9} Z Y^{28}+
& X^{1} Z Y^{29}+X^{2} Z Y^{29}+X^{3} Z Y^{29}+X^{4} Z Y^{29}+X^{6} Z Y^{29}+X^{7} Z Y^{29}+X^{9} Z Y^{29}+\\ X^{1} Z Y^{30}+
& X^{2} Z Y^{30}+X^{3} Z Y^{30}+X^{4} Z Y^{30}+X^{6} Z Y^{30}+X^{7} Z Y^{30}+X^{9} Z Y^{30}+X^{1} Z Y^{31}+\\ X^{2} Z Y^{31}+
& X^{3} Z Y^{31}+X^{4} Z Y^{31}+X^{6} Z Y^{31}+X^{7} Z Y^{31}+X^{9} Z Y^{31}+X^{1} Z Y^{32}+X^{2} Z Y^{32}+\\ X^{3} Z Y^{32}+
& X^{4} Z Y^{32}+X^{6} Z Y^{32}+X^{7} Z Y^{32}+X^{9} Z Y^{32}+X^{1} Z Y^{33}+X^{2} Z Y^{33}+X^{3} Z Y^{33}+\\ X^{4} Z Y^{33}+
& X^{6} Z Y^{33}+X^{7} Z Y^{33}+X^{9} Z Y^{33}+X^{1} Z Y^{34}+X^{2} Z Y^{34}+X^{3} Z Y^{34}+X^{4} Z Y^{34}+\\ X^{6} Z Y^{34}+
& X^{7} Z Y^{34}+X^{9} Z Y^{34}+X^{1} Z Y^{35}+X^{2} Z Y^{35}+X^{3} Z Y^{35}+X^{4} Z Y^{35}+X^{6} Z Y^{35}+\\ X^{7} Z Y^{35}+
& X^{9} Z Y^{35}+X^{1} Z Y^{36}+X^{2} Z Y^{36}+X^{3} Z Y^{36}+X^{4} Z Y^{36}+X^{6} Z Y^{36}+X^{7} Z Y^{36}+\\ X^{9} Z Y^{36}+
& X^{1} Z Y^{37}+X^{2} Z Y^{37}+X^{3} Z Y^{37}+X^{4} Z Y^{37}+X^{6} Z Y^{37}+X^{7} Z Y^{37}+X^{9} Z Y^{37}+\\ X^{1} Z Y^{38}+
& X^{2} Z Y^{38}+X^{3} Z Y^{38}+X^{4} Z Y^{38}+X^{6} Z Y^{38}+X^{7} Z Y^{38}+X^{9} Z Y^{38}+X^{1} Z Y^{39}+\\ X^{2} Z Y^{39}+
& X^{3} Z Y^{39}+X^{4} Z Y^{39}+X^{6} Z Y^{39}+X^{7} Z Y^{39}+X^{9} Z Y^{39}+X^{1} Z Y^{40}+X^{2} Z Y^{40}+\\ X^{3} Z Y^{40}+
& X^{4} Z Y^{40}+X^{6} Z Y^{40}+X^{7} Z Y^{40}+X^{9} Z Y^{40}+X^{1} Z Y^{41}+X^{2} Z Y^{41}+X^{3} Z Y^{41}+\\ X^{4} Z Y^{41}+
& X^{6} Z Y^{41}+X^{7} Z Y^{41}+X^{9} Z Y^{41}+X^{1} Z Y^{42}+X^{2} Z Y^{42}+X^{3} Z Y^{42}+X^{4} Z Y^{42}+\\ X^{6} Z Y^{42}+
& X^{7} Z Y^{42}+X^{9} Z Y^{42}+X^{1} Z Y^{43}+X^{2} Z Y^{43}+X^{3} Z Y^{43}+X^{4} Z Y^{43}+X^{6} Z Y^{43}+\\ X^{7} Z Y^{43}+
& X^{9} Z Y^{43}+X^{1} Z Y^{44}+X^{2} Z Y^{44}+X^{3} Z Y^{44}+X^{4} Z Y^{44}+X^{6} Z Y^{44}+X^{7} Z Y^{44}+\\ X^{9} Z Y^{44}+
& X^{1} Z Y^{45}+X^{2} Z Y^{45}+X^{3} Z Y^{45}+X^{4} Z Y^{45}+X^{6} Z Y^{45}+X^{7} Z Y^{45}+X^{9} Z Y^{45}+\\ X^{1} Z Y^{46}+
& X^{2} Z Y^{46}+X^{3} Z Y^{46}+X^{4} Z Y^{46}+X^{6} Z Y^{46}+X^{7} Z Y^{46}+X^{9} Z Y^{46}+X^{1} Z Y^{47}+\\ X^{2} Z Y^{47}+
& X^{3} Z Y^{47}+X^{4} Z Y^{47}+X^{6} Z Y^{47}+X^{7} Z Y^{47}+X^{9} Z Y^{47}+X^{1} Z Y^{48}+X^{2} Z Y^{48}+\\ X^{3} Z Y^{48}+
& X^{4} Z Y^{48}+X^{6} Z Y^{48}+X^{7} Z Y^{48}+X^{9} Z Y^{48}+X^{1} Z Y^{49}+X^{2} Z Y^{49}+X^{3} Z Y^{49}+\\ X^{4} Z Y^{49}+
& X^{6} Z Y^{49}+X^{7} Z Y^{49}+X^{9} Z Y^{49}+X^{1} Z Y^{50}+X^{2} Z Y^{50}+X^{3} Z Y^{50}+X^{4} Z Y^{50}+\\ X^{6} Z Y^{50}+
& X^{7} Z Y^{50}+X^{9} Z Y^{50}
\end{split}
\end{equation}

\end{example}

\subsection{$R$ is isomorphic to a matrix ring over a finite field}

\medskip

In case $R\cong \rm{Mat}_n(\mathbb{F}_q)$, for some finite field $\mathbb{F}_q$ and some $n \geq 1$, we can apply Caley-Hamilton Theorem and then reduce the computation of the shared key to a linear algebra problem on $\mathbb{F}_q$. This is based on the fact that there exist a ring homomorphism $\rm{Mat}_m(\it{R})\cong \rm{Mat}_{m\times n}(\mathbb{F}_q)$. 

Now, given $S, M_1, M_2, A$ and $B$, we can solve the system of equations given by 

$$\sum_{i=0}^{k-1} \sum_{j=0}^{k-1} d_{i,j} M_{1}^{i} S M_{2}^{j} = A$$

\noindent getting the coefficients $d_{i,j}$ and we can build the function $$F[X,Y,Z] = \sum_{i=0}^{k-1} \sum_{j=0}^{k-1} d_{i,j} X^{i} Y Z^{j}$$

Then 

\begin{equation}
    \begin{split}
        F[M_{1},M_{2},B] & = \sum_{i =0}^{k-1} \sum_{j =0}^{k-1} d_{i,j} M_{1}^i B M_{2}^j = \\ & = \sum_{i =0}^{k-1} \sum_{j =0}^{k-1} d_{i,j} M_{1}^i p_{b}(M_{1}) S q_{b}(M_{2}) M_{2}^j = \\ & = \sum_{i =0}^{k-1} \sum_{j =0}^{k-1} d_{i,j} p_{b}(M_{1}) M_{1}^i  S  M_{2}^j q_{b}(M_{2})= \\ & = p_{b}(M_{1})\sum_{i =0}^{k-1} \sum_{j =0}^{k-1} d_{i,j} M_{1}^i S M_{2}^j  q_{b}(M_{2}) = \\
& = p_{b}(M_{1}) F[M_{1},M_{2},S] q_{b}(M_{2}) = \\ & = p_{b}(M_{1}) A q_{b}(M_{2})
    \end{split}
\end{equation}

\noindent which is the shared key and that is clearly obtained in polynomial time.

\subsection{Other cases}

Taking into account Theorem \ref{necessary}, it remains to analyze three cases. 

In case $R$ is a zero multiplication ring of prime order, then the product of matrices through the protocol provide zero matrices. Thus, this does not have any sense in this environment. The same occurs in the case $R$ contains an absorbing element $\infty$, i.e. $x\cdot \infty=\infty \cdot x=\infty$ for every $x\in R$ and $R+R=\{ \infty \}$. In this case, the matrices involved will be formed by $\infty$ in every entrance. 

We just have to analyze the case $R$ has order 2. Then by Theorem \ref{order2} we have that if $R$ is congruence-simple additively commutative, then it is isomorphic to one of those semirings. We can observe that $T_3, T_4, T_5$ and $T_6$ are additively idempotent and so the preceding attack for such semirings applies. In the cases $T_1, T_2$ and $T_7$, by their structure, they will provide zero matrices and therefore the protocol has no sense. Finally, $T_8$ is exactly $\mathbb{Z}_2$ and then, the reasoning over finite fields applies. 

It remains the case $R$ is additively non-commutative. But by Lemma \ref{order2idempotent}, then $R$ is additively idempotent, and so we are again in a one of the previous analyzed cases.

\section{Conclusions}

 We have shown the existence of arguments that provide a cryptanalysis for the key exchange protocol taking a congruence-simple semiring as an environment as it is proposed in \cite{maze}. We provide a new algorithm that reveals the shared key in the case of additively idempotent congruence-simple semirings and give arguments in the rest of the cases of congruence-simple rings. Moreover, given that the existence of any congruence in such semirings, as it is stated by the own authors, would provide a Pohlig-Hellman type reduction of the corresponding Semigroup Acton Problem, then we can conclude that this

\end{document}